\documentclass[1p,times,sort&compress]{elsarticle}
\usepackage[utf8]{inputenc}
\usepackage[T1]{fontenc}
\usepackage{amsmath,amsfonts,amsthm,fixltx2e,tikz}

\DeclareMathOperator{\Ff}{Fact}
\DeclareMathOperator{\Pre}{Pref}
\DeclareMathOperator{\Su}{Suff}
\newcommand{\Nn}{\mathbb N}
\newcommand{\Ww}{\mathbf w}
\newcommand{\Ss}{\mathbf s}
\newcommand{\Cc}{\mathbf c}
\newcommand{\eps}{\varepsilon}
\newcommand{\wt}{\widetilde}

\newtheorem{thm}{Theorem}[section]
\newtheorem{prop}[thm]{Proposition}
\newtheorem{lemma}[thm]{Lemma}
\newtheorem{cor}[thm]{Corollary}
\theoremstyle{definition}
\newtheorem{example}[thm]{Example}
\newtheorem{remark}[thm]{Remark}

\begin{document}
\begin{frontmatter}
\title{Reversible Christoffel factorizations}
\author[tku]{Michelangelo Bucci}
\ead{michelangelo.bucci@utu.fi}

\author[nap]{Alessandro De Luca\corref{cor1}}
\ead{alessandro.deluca@unina.it}

\author[tku,lys]{Luca Q. Zamboni}
\ead{lupastis@gmail.com}

\cortext[cor1]{Corresponding author}

\address[tku]{Department of Mathematics, University of Turku, FI-20014 Turku, Finland}
\address[nap]{DIETI, Universit\`a degli Studi di Napoli Federico II\\via Claudio 21, 80125 Napoli, Italy}
\address[lys]{Universit\'e de Lyon, Universit\'e Lyon 1, CNRS UMR 5208 Institut Camille Jordan\\B\^atiment du Doyen Jean Braconnier, 43, blvd du 11 novembre 1918\\ F-69622 Villeurbanne Cedex, France}
\journal{Theoretical Computer Science}
\begin{abstract}
We define a family of natural decompositions of Sturmian words in Christoffel words, called \emph{reversible Christoffel} (RC) factorizations. They arise from the observation that two Sturmian words with the same language have (almost always) arbitrarily long Abelian equivalent prefixes. Using the three gap theorem
, we prove that in each RC factorization, only 2 or 3 distinct Christoffel words may occur. We begin the study of such factorizations, considered as infinite words over 2 or 3 letters, and show that in the general case they are either Sturmian words, or obtained by a three-interval exchange transformation.
\end{abstract}

\end{frontmatter}

\section{Introduction}
In combinatorics on words and symbolic dynamics, it is often meaningful to look at two infinite words $\Ww ,\Ww' $ and determine the segments where they coincide, that is, locate maximal occurrences of factors $u$ such that $\Ww =pu\Ss $, $\Ww' =p'u\Ss'$ for some words $p$ and $p'$ of equal length and some infinite words $\Ss ,\Ss'$.

If $\Ww $ and $\Ww'$ are two fixed points of an irreducible Pisot substitution
$\varphi$, the \emph{strong coincidence conjecture} (proved by~\citet{bardi} in the binary case) states that there exists a letter $a$ and two factorizations
\begin{equation}
\label{eq:pas}
\Ww =pa\Ss ,\quad \Ww' =p'a\Ss'
\end{equation}
such that $p$ and $p'$ are Abelian equivalent, i.e., an ``anagram'' of each other. This has two remarkable consequences:
\begin{enumerate}
	\item\label{prox} $\Ww $ and $\Ww'$ agree on arbitrarily long segments (as defined above), i.e., they are \emph{proximal},
	\item\label{Abcom} $\Ww $ and $\Ww'$ have arbitrarily long Abelian equivalent prefixes; in short, we say that they are \emph{Abelian comparable}.
\end{enumerate}
Indeed, from~\eqref{eq:pas} it follows
$\Ww =\varphi^{n}(p)\varphi^{n}(a)\varphi^{n}(\Ss )$ and
$\Ww' =\varphi^{n}(p')\varphi^{n}(a)\varphi^{n}(\Ss')$ for all $n\geq 0$.

This induces two factorizations (\emph{comparison})
\begin{equation}
\label{eq:Abcom}
\begin{split}
\Ww\phantom' &=x_{1}x_{2}\cdots x_{n}\cdots\\
\Ww '&=x_{1}'x_{2}'\cdots x_{n}'\cdots\\
\end{split}
\end{equation}
defined so that each pair of Abelian equivalent prefixes of $\Ww $ and $\Ww'$ is
$(x_{1}\cdots x_{k},x_{1}'\cdots x_{k}')$ for some $k\geq 0$; equivalently, for all $k>0$, $x_{k}$ and $x_{k}'$ are the shortest nonempty Abelian equivalent prefixes of the infinite words $x_{k}x_{k+1}\cdots$ and $x_{k}'x_{k+1}'\cdots$. 

\bigskip
In this paper (a preliminary version of which was presented at the first RuFiDiM~\cite{rufidim}), we look at Sturmian words over $A=\{0,1\}$ from a similar point of view. We recall that an infinite word over $A$ is \emph{Sturmian} if it has exactly $n+1$ distinct factors of each length $n\geq 0$. The first systematic study of Sturmian sequences is usually credited to \citet{mhsturm}, from the point of view of symbolic dynamics.
In fact, every Sturmian word can be realized either as a \emph{lower mechanical word} or as an \emph{upper} one. The lower (resp.~upper) mechanical word $\Ss _{\alpha,\rho}$ (resp.~$\Ss _{\alpha,\rho}'$) of \emph{slope} $\alpha$ and \emph{intercept} $\rho$, with 
$0\leq\alpha,\rho<1$, is the infinite word indexed over $\Nn$ whose
$n$-th letter is 0 if 
\[\{n\alpha+\rho\}<1-\alpha\quad
\text{ (resp.~if }0<\{n\alpha+\rho\}\leq 1-\alpha\text{)}\]
and 1 otherwise (denoting by $\{\sigma\}$ the \emph{fractional part} $\sigma-\lfloor\sigma\rfloor$ of the real number $\sigma$). Thus mechanical words encode \emph{rotations} by angle $2\pi\alpha$ on a circle, and 
$\alpha$ gives the frequency of the letter $1$ in the infinite word. Moreover, the slope determines the \emph{language} (set of factors).
Among binary words, mechanical words are characterized by the \emph{balance} property: the number of occurrences of the letter $1$ in two factors of the same length may differ at most by 1.

When $\alpha\notin\mathbb Q$, the word
$\Ss _{\alpha,\alpha}=\Ss _{\alpha,\alpha}'=:\Cc $ is said to be the \emph{characteristic} (or \emph{standard}) Sturmian word of slope $\alpha$. Its prefixes are exactly all \emph{left special} factors of the  Sturmian words of slope $\alpha$, i.e., $p\in\Pre(\Cc )$ if and only if $0p,1p\in\Ff(\Ss_{\alpha,\rho})$ for any $\rho$.
We say that a Sturmian word is \emph{singular} if it contains the characteristic word (of the same slope) as a proper suffix.
By definition, we have
\[\Ss _{\alpha,0}=0\Cc\quad\text{ and }\quad\Ss_{\alpha,0}'=1\Cc \;;\]
any other singular Sturmian word can be written as $\tilde p01\Cc $ or $\tilde p10\Cc $, where $p\in\Pre(\Cc )$. It is easy to see that every nonsingular Sturmian word is both an upper mechanical word and a lower one.

Mechanical words of irrational slope are exactly all Sturmian words, whereas those of rational slope are periodic words; when $\alpha\in\mathbb Q$, the shortest $v$ such that $\Ss _{\alpha,0}=v^{\omega}$ 
(resp.~$\Ss _{\alpha,0}'=v^{\omega}$) is the \emph{lower} (resp.~\emph{upper}) \emph{Christoffel word} of slope $\alpha$. It is well-known (cf.~\cite{bersdelu}) that the set of lower (resp.~upper) Christoffel words can be characterized as
$A\cup 0\mathcal{P} 1$ (resp.~$A\cup 1\mathcal P 0$), where 
$\mathcal P$ is the set of \emph{central words}, i.e., words $u$ such that both 
$0u1$ and $1u0$ are factor of some Sturmian word $s$. Central words are exactly all palindromic prefixes of characteristic Sturmian words. 
The following characterization of central words is well-known.
\begin{prop}[See \citet{delustr}]
A word $u\in A^{*}$ is central if and only if it is a \emph{palindrome} satisfying \emph{one} of the following conditions:
\begin{enumerate}
	\item $u$ is a power of a letter, i.e., $u\in 0^{*}\cup 1^{*}$, or
	\item $u=p01q=q10p$ for some palindromes $p,q\in A^{*}$.
\end{enumerate}
In the latter case, the words $p,q$ are central too, and uniquely determined; one of them is the longest palindromic prefix (and suffix) of $u$. 
\end{prop}
By iterated application of the previous result, one easily obtains:
\begin{cor}
\label{thm:cent}
Let $v$ be a proper palindromic prefix of a central word $u$, such that
$v\notin 0^{*}\cup 1^{*}$. Then either $v01$ or $v10$ is a prefix of $u$.
\end{cor}
Since central words are palindromes, we have that $v=0u1$ is a lower Christoffel word if and only if its \emph{reversal} $\tilde v=1u0$ is an upper Christoffel word. We recall that any nontrivial (i.e., longer than a letter) Christoffel word can be uniquely written as a product of two (shorter) Christoffel words. All pairs $(u,v)$ of lower Christoffel words such that $uv$ is Christoffel make up the \emph{Christoffel tree} (cf.~\citet{bersdelu}; see Figure~\ref{fig:chr}) where the pair $(0,1)$ is the root%
\footnote{Using $(1,0)$ instead, all \emph{upper} Christoffel pairs are obtained.}, and every node $(u,v)$ has the two children $(u,uv)$ and $(uv,v)$. Moreover, all lower Christoffel factors of an infinite Sturmian word are found on an infinite path on the tree.
\begin{figure}[hbt]
\centering
\begin{tikzpicture} [level distance=12mm, 
level/.style={sibling distance=120mm/(2^#1)}, 
level 4/.style={level distance=10mm}
]
\node {$(0,1)$}
	child {node {$(0,01)$}
		child {node {$(0,001)$}
			child {node {$(0,0001)$}
				child {node {\vdots}}
				child {node {\vdots}}
			}
			child {node {$(0001,001)$}
				child {node {\vdots}}
				child {node {\vdots}}
			}
		}
		child {node {$(001,01)$}
			child 
				{node {\vdots}}
			child {node {$(00101,01)$}
				child {node {\vdots}}
				child {node {\vdots}}
			}
		}
	}
	child {node {$(01,1)$}
		child {node {$(01,011)$}
			child 
				{node {\vdots}}
			child {node {$(01011,011)$}
				child {node {\vdots}}
				child {node {\vdots}}
			}
		}
		child {node {$(011,1)$}
			child 
				{node {\vdots}}
			child {node {$(0111,1)$}
				child {node {\vdots}}
				child {node {\vdots}}
			}
		}
	};
\end{tikzpicture}
\label{fig:chr}
\caption{The Christoffel tree}
\end{figure}
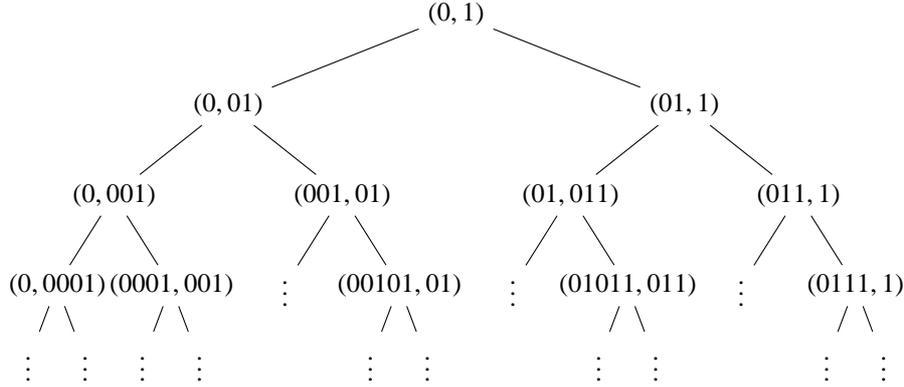
Thus, in particular, the following properties hold:
\begin{prop}
\label{thm:cmix}
Let $u,v,z$ be Christoffel words such that $z=uv$, and let $\Ww$ be a Sturmian word with $z\in\Ff(\Ww)$. Then:
\begin{enumerate}
\item $u^2v$ and $uv^{2}$ are Christoffel words too, and exactly one of them is a factor of $\Ww$;
\item if $\{u,v\}\neq\{0,1\}$, then either $u$ is a prefix of $v$, or $v$ is a suffix of $u$.
\end{enumerate}
\end{prop}

For more information on Sturmian and Christoffel words, we refer the reader to~\cite{ch2acow,blrs}.

\medskip
In the next sections, we shall deal with the comparison of Sturmian words. We begin by proving (Proposition~\ref{thm:abcom}) that all pairs of Sturmian words of the same slope, except one, are Abelian comparable.
Our main result (Theorem~\ref{thm:main}) shows that 
for Abelian comparable Sturmian words $\Ww$ and $\Ww'$, the factorizations in~\eqref{eq:Abcom} have at most 3 distinct terms, i.e.,
the sets $\{x_{i}\}_{i>0}$ and $\{x_{i}'\}_{i>0}$ have cardinality 2 or 3; furthermore, all such terms are Christoffel words.
Finally, we shall examine the structure of these factorizations.

\section{Comparison of Sturmian words: RC factorizations}

Trivially, if two Sturmian words are Abelian comparable then they have the same slope (and hence the same language).
The following proposition shows that the converse holds too, with a single exception.
\begin{prop}
\label{thm:abcom}
Let $\Ww ,\Ww '$ be Sturmian words of slope $\alpha$, and $\Cc =\Ss _{\alpha,\alpha}$ be the characteristic word. If $\{\Ww ,\Ww '\}\neq\{0\Cc ,1\Cc \}$, then $\Ww $ and $\Ww'$ are Abelian comparable.
\end{prop}
\begin{proof}
Let $x_{1}$ and $x_{1}'$ be the shortest nonempty prefixes (of $\Ww $ and $\Ww'$ respectively) which are Abelian equivalent. These are well defined; in fact, suppose by contradiction that $\Ww $ and $\Ww'$ have no Abelian equivalent prefixes except $\eps$. By the balance property, it follows $\{\Ww,\Ww'\}=\{0\mathbf t,1\mathbf t\}$ for some infinite word $\mathbf t$; as all prefixes of $\mathbf t$ are left special, we get $\mathbf t=\Cc$, contradicting our hypothesis.

Let then $\Ww=x_{1}\Ww_{1}$ and $\Ww'=x_{1}'\Ww_{1}'$ for some Sturmian words 
$\Ww_{1},\Ww_{1}'$ having the same language as $\Ww $ and $\Ww'$. We have
$\{\Ww_{1},\Ww_{1}'\}\neq\{0\Cc ,1\Cc \}$, for otherwise we would obtain $\{\Ww,\Ww'\}=\{\tilde p10\Cc ,\tilde p01\Cc \}$ for some $p\in\Pre(\Cc)$, and then $\{x_{1},x_{1}'\}=\{\tilde p1,\tilde p0\}$, which is absurd as $x_{1}$ and $x_{1}'$ are Abelian equivalent.
Hence we can iterate this argument to get infinitely many Abelian equivalent prefixes of $\Ww $ and $\Ww'$.
\end{proof}

\begin{example}
Let $\alpha=(3-\sqrt{5})/2$. The \emph{Fibonacci word}
\[\mathbf f=\Ss_{\alpha,\alpha}=010010100100101001010010010100100101001010010010100101001001\cdots\]
is the most famous Sturmian word. The Abelian comparison of the words $\mathbf f$ and $\mathbf{f'}:=\Ss_{\alpha,4/5}$ is:
\[
\begin{array}%
{@{\extracolsep{-0.2em}}rl|c|c|c|c|c|c|c|c|c|c|c|c|c|c|c|c|c|c|c|c}
\mathbf f\phantom' &=
01 &0 &01 &01 &0 &01 &001 &01 &0 &01 &01 &0 &01 &0 &01 &01 &0 &01 &001 &01 &\cdots\\ \mathbf{f'} &=
10 &0 &10 &10 &0 &10 &100 &10 &0 &10 &10 &0 &10 &0 &10 &10 &0 &10 &100 &10 &\cdots\end{array}
\]
\end{example}

\medskip
Let $\Ww,\Ww'$ be two Sturmian words having the same language, and suppose
$\{\Ww,\Ww'\}\neq\{0\Cc ,1\Cc \}$ where $\Cc$ is the characteristic word with the same slope.
By Proposition~\ref{thm:abcom}, $\Ww $ and $\Ww'$ are then Abelian comparable; let their comparison be given by~\eqref{eq:Abcom}.

By definition, $x_{i}$ and $x_{i}'$ are Abelian equivalent for all $i\geq 1$. By the balance property, it follows either $x_{i}=x_{i}'\in A$, or $\{x_{i},x_{i}'\}=\{0u1,1u0\}$ for some factor $u$ of $\Ww $, which is then a central word. We conclude that in all cases, $x_{i}$ and $x_{i}'$ are Christoffel words, with $x_{i}'=\wt{x_{i}}$. Thus we can write:
\begin{equation}
\label{eq:RC}
\begin{split}
\Ww\phantom' &=x_{1}x_{2}\cdots x_{n}\cdots,\\
\Ww' &=\wt{x_{1}}\wt{x_{2}}\cdots \wt{x_{n}}\cdots\;.
\end{split}
\end{equation}

Conversely, if $(x_{n})_{n>0}$ is a sequence of Christoffel words such that both infinite words in~\eqref{eq:RC} are Sturmian, then the Abelian comparison of $\Ww $ and $\Ww'$ yields exactly the same factorizations.

This motivates the following definition: we call \emph{reversible Christoffel (RC) factorization} of a Sturmian word $\Ww $ any infinite sequence $(x_{k})_{k>0}$ of Christoffel words such that
\begin{enumerate}
\item $\Ww=x_{1}x_{2}\cdots x_{n}\cdots$, and
\item $\Ww':=\wt{x_{1}}\wt{x_{2}}\cdots\wt{x_{n}}\cdots$ is a Sturmian word.
\end{enumerate}
We also say that $(x_{k})_{k>0}$ is \emph{the} RC factorization of $\Ww $ \emph{determined by} $\Ww'$.

A trivial RC factorization is obtained by choosing all $x_{k}$'s to be single letters, so that $\Ww'=\Ww$. The definition implies that every choice of $\Ww'$ determines a distinct factorization of $\Ww $; this proves the following statement.

\begin{prop}
Every Sturmian word admits uncountably many distinct RC factorizations.
\end{prop}

The following result is an immediate consequence of the balance property:
\begin{prop}
\label{thm:uplow}
Let $\Ww $ be a Sturmian word, and $(x_{k})_{k>0}$ be any RC factorization of $\Ww $. Then the terms $x_{k}$, $k\geq 1$, are either \emph{all} upper Christoffel words, or they are all lower Christoffel words.
\end{prop}

In the remainder of this section, we shall use the above definitions and results to give a proof of a stronger (and probably known) version of the strong coincidence conjecture in the case of Sturmian words, namely Proposition~\ref{thm:stcoSt} below. We need the following lemma (a restatement of \cite[Lemma~5.2]{BuPuZa}):
\begin{lemma}[See~\citet{BuPuZa}]
\label{thm:proxSt}
Two distinct Sturmian words $\Ww,\Ww'$ are proximal if and only if they can be written as $\Ww=q\Cc$, $\Ww'=q'\Cc$ for some $q,q'$ with $|q|=|q'|$; that is, if and only if they contain the characteristic word at the same position.
\end{lemma}
As observed above, in such a case the set $\{q,q'\}$ is either $\{0,1\}$ or $\{\tilde p01,\tilde p10\}$ for some $p\in\Pre(\Cc)$.

We recall that a morphism $f:A^{*}\to A^{*}$ is said to be \emph{Sturmian} if it maps Sturmian words to Sturmian words; as is well-known~(cf.~\cite{ch2acow}), for a morphism to be Sturmian it suffices to map \emph{one} Sturmian word to another one. Furthermore, Sturmian morphisms form a monoid generated by the three substitutions
\begin{equation}
\label{eq:stmor}
E:0\mapsto 1,1\mapsto 0,\quad \varphi:0\mapsto 01,1\mapsto 0,\quad\text{ and }\tilde\varphi:0\mapsto10,1\mapsto 0\;.
\end{equation}
From this, an easy induction argument gives the following property, also well-known:
\begin{lemma}
\label{thm:f0f1}
If $f\notin\{\mathrm{id},E\}$ is a Sturmian morphism, then one between $f(0)$ and $f(1)$ is  a proper prefix or a proper suffix of the other.
\end{lemma}

We also recall that a complete characterization of Sturmian fixed points of morphisms in terms of slope and intercept was given by~\citet{yasutomi} (see also~\cite{BeEiItRa}).

The following result about coincidence for Sturmian words 
can be easily proved as a direct consequence of well-known characterizations and properties of Sturmian morphisms. We give here a proof based on RC factorizations:
\begin{prop}
\label{thm:stcoSt}
Let $\Ww,\Ww'$ be distinct Sturmian words that are fixed points of a nontrivial morphism $f:A^{*}\to A^{*}$. Then $\{\Ww,\Ww'\}=\{01\Cc,10\Cc\}$, where $\Cc$ is the characteristic word having the same language.
\end{prop}
\begin{proof}
It is easy to see that $\Ww$ and $\Ww'$ have the same language. Without loss of generality, we can assume that $\Ww$ starts with $0$ and $\Ww'$ starts with $1$. Let us first show that the combination $\Ww=0\Cc,\,\Ww'=1\Cc$ is not possible. If $f(0\Cc)=0\Cc$ and $f(1\Cc)=1\Cc$, it would follow $\Cc=\lambda f(\Cc)$ for some word $\lambda$, and then $|f(0)|=|f(1)|$, since $\Cc$ is not ultimately periodic. This is impossible in view of Lemma~\ref{thm:f0f1}.

Thus by Proposition~\ref{thm:abcom}, $\Ww$ and $\Ww'$ are Abelian comparable; let $\Ww=x_{1}x_{2}\cdots$ and $\Ww'=\wt{x_{1}}\wt{x_{2}}\cdots$ be their corresponding RC factorizations. By contradiction, suppose they contain only words of length $>1$. By Proposition~\ref{thm:uplow}, for all $n\geq 1$ we can write $x_{n}=0u_{n}1$ for suitable words $u_{n}$. Since $f(x_{1})$ is Abelian equivalent to $f(\wt{x_{1}})$, there must be an $m>1$ such that $f(x_{1})=x_{1}x_{2}\cdots x_{m}$. Hence $x_{1}$ and its image both end with $1$, so that the word $f(1)$, which clearly starts with $1$, ends in $1$ as well. As $f(\wt{x_{1}})=\wt{x_{1}}\cdots\wt{x_{m}}$, by the same argument it follows that $f(0)$ begins and ends in $0$. Again, this contradicts Lemma~\ref{thm:f0f1}.

Therefore, there must be some term $x_{i}$ of length 1 in the RC factorization, that is,~\eqref{eq:pas} holds.
As discussed above, this implies that $\Ww$ and $\Ww'$ are proximal. Since $\{\Ww,\Ww'\}\neq\{0\Cc,1\Cc\}$, by Lemma~\ref{thm:proxSt} it follows $\{\Ww,\Ww'\}=\{\tilde p01\Cc,\tilde p10\Cc\}$ for some $p\in\Pre(\Cc)$; as $\Ww$ starts with $0$ and $\Ww'$ with 1, the assertion is proved.
\end{proof}
\begin{example}
The words $01\mathbf f$ and $10\mathbf f$ are both fixed by $\tilde\varphi^{2}$, as defined in~\eqref{eq:stmor}.
\end{example}

\section{Main results}
\subsection{Terms of RC factorizations}
We recall the following well-known result by~\citet{slater}, deeply related to the \emph{three distance theorem} proved by~\citet{sosd} (see also~\cite{ab3d}):
\begin{thm}[Three gap theorem]
\label{thm:3gap}
Let $\alpha$ be an irrational number with $0<\alpha<1$, and let $0<\beta<1/2$. The gaps between the successive integers $n$ such that $\{n\alpha\} <\beta$  take either two or three values, one being the sum of the other two.
\end{thm}

Our main theorem shows that RC factorizations have at most 3 distinct terms.
\begin{thm}
Let $\Ww $ be a Sturmian word, and $\Ww=x_{1}x_{2}\cdots x_{n}\cdots$ be an RC factorization of $\Ww $. The cardinality of the set $X=\{x_{n}\mid n>0\}$ is either 2 or 3, and in the latter case, the longest element of $X$ is obtained by concatenating the other two.
\label{thm:main}
\end{thm}
\begin{proof} 
Since Sturmian words are not periodic, the set $X$ has cardinality at least two. 
Let $\Ww'=\wt{x_{1}}\wt{x_{2}}\cdots$, and suppose first that $\Ww $ and $\Ww'$ are both lower mechanical words, so that $\Ww=\Ss_{\alpha,\rho}$ and
$\Ww'=\Ss_{\alpha,\rho'}$ for some $\rho,\rho'\in [0,1[$. Without loss of generality, we may suppose $\alpha<1/2$ (otherwise it suffices to exchange the roles of the letters 0 and 1), and $\beta:=\{\rho'-\rho\}\leq 1/2$ (otherwise we swap $\Ww $ and $\Ww'$). Hence $\alpha<1-\alpha$ and $\beta\leq 1-\beta$.

We distinguish two possibilities:
\begin{enumerate}[C{a}se 1.]
\item If $\alpha\leq\beta$, then $\Ww $ and $\Ww'$ cannot be 1 at the same time, i.e., there is no $n$ for which $\{n\alpha+\rho\}$ and $\{n\alpha+\rho'\}$ are both larger than $1-\alpha$. Assuming, in view of Proposition~\ref{thm:uplow}, that all terms $x_{k}$ ($k\geq 1$) are lower Christoffel words (the ``upper'' case being similar), this implies $X\subseteq\{0\}\cup\{0^{k}1\mid k>0\}$, since any other lower Christoffel word $x$ would have a 1 in the same position as in $\tilde x$. Let $i,j,k_{1}$, and $k_{2}$ be positive integers such that $i<j$, 
$x_{i}=0^{k_{1}}1$, and $x_{j}=0^{k_{2}}1$. This implies, since $\Ww $ and $\Ww'$ have the same language, that $\Ww $ has the two factors
\begin{eqnarray*}
x_{i}x_{i+1}\cdots x_{j-1}x_{j}&=&0^{k_{1}}1x_{i+1}\cdots x_{j-1}0^{k_{2}}1\,,
\;\text{ and}\\
\wt{x_{i}}\wt{x_{i+1}}\cdots \wt{x_{j-1}}\wt{x_{j}}&=&10^{k_{1}}\wt{x_{i+1}}\cdots \wt{x_{j-1}}10^{k_{2}}
\end{eqnarray*}
so that $|k_{1}-k_{2}|\leq 1$ as a consequence of the balance property. By the arbitrary choice of $i$ and $j$, it follows $X\subseteq\{0,0^{h}1,0^{h+1}1\}$ for some $h>0$, which settles this case.

\item If $\alpha>\beta$, $\Ww $ and $\Ww'$ differ exactly in all positions $n$ such that $\{n\alpha+\rho\}\in I_{1}\cup I_{2}$, with 
$I_{1}=[1-\alpha-\beta,1-\alpha[$ and $I_{2}=[1-\beta,1[$. Note that if 
$\{n\alpha+\rho\}\in I_{1}$, then $\{(n+1)\alpha+\rho\}\in I_{2}$. Hence, if
$\rho\notin I_{2}$ we derive $X\subseteq\{0,1,01\}$.

If $\rho\in I_{2}$, let 
$(n_{i})_{i\geq 0}$ be the increasing sequence of all positive integers such that $\{n_{i}\alpha+\rho\}\in I_{2}$. For all $i\geq 0$ we have $\{n_{i}\alpha+\rho+\beta\}<\beta$, so that by Theorem~\ref{thm:3gap} it follows that the set
$\{n_{i+1}-n_{i}\mid i\geq 0\}$ is contained in $\{k_{1},k_{2},k_{1}+k_{2}\}$ for some distinct integers $k_{1},k_{2}$ (both greater than 1, as $\alpha<1-\beta$). For all $i\geq 0$, $x_{i}$ is the factor of $\Ww$ starting with the $n_{i}$-th letter and ending in the $(n_{i+1}-1)$-th one, since it can be written as $1u0$ for some $u$ such that the corresponding factor of $\Ww'$ is $0u1$. Thus, the elements of $X$ may have length $k_{1}$, $k_{2}$, or $k_{1}+k_{2}$. There cannot be two of the same length, since 
any factor of $\Ww$ which is an upper Christoffel word of length~$\geq 2$ can be written as $1u0$ where $u$ is a palindromic prefix of the characteristic word of slope $\alpha$.
Hence $X$ has cardinality 2 or 3.

Let now $X=\{y_{1},y_{2},z\}$ with $|z|=k_{1}+k_{2}$ and $|y_{i}|=k_{i}$ for $i=1,2$; we can write $z=1v0$ and $y_{i}=1u_{i}0$. Since $u_{1}$, $u_{2}$, and $v$ are all palindromic prefixes of the characteristic word, $u_i$ is a prefix and a suffix of $v$ for $i=1,2$. Hence
\begin{equation}
\label{eq:cent}
v=u_{1}abu_{2}=u_{2}bau_{1}
\end{equation} for some letters $a,b\in\{0,1\}$. If $a\neq b$, it follows either $z=y_{1}y_{2}$ or $z=y_{2}y_{1}$ and we are done.

By contradiction, let us then suppose $a=b$. From $\alpha<1/2$ we derive that $a=b=0$, and that both $u_{1}$ and $u_{2}$ start with the letter 0. If both were a power of 0 we would reach a contradiction, as no point of the (dense) sequence $\{n\alpha+\rho\}$ would lie in the nonempty interval 
$[1-\alpha,1-\beta[$. Hence there exists $j\in\{1,2\}$ such that $u_{j}$ contains both 0 and 1, so that by Corollary~\ref{thm:cent} either $u_{j}01$ or $u_{j}10$ is a prefix of $v$. This is a contradiction, since we are assuming~\eqref{eq:cent} with $a=b=0$.
\end{enumerate}
When $\Ww $ and $\Ww'$ are both upper mechanical words, the proof is symmetrical. 

Let us now suppose that both $\Ww$ and $\Ww'$ are singular, and that only one of them is a lower mechanical word. By contradiction, suppose that
$X$ has more than 3 elements. This means that there exists $j>3$ such that the set $\{x_{1},x_{2},\ldots,x_{j}\}\subseteq X$ has cardinality 4. The word 
$r=x_{1}\cdots x_{j}$ is prefix of infinitely many nonsingular Sturmian words of slope $\alpha$, and so is $r':=\wt{x_{1}}\cdots\wt{x_{j}}$. Let $\mathbf t,\mathbf t'$ be two such nonsingular extensions, of $r$ and $r'$ respectively. The RC factorization of $\mathbf t$ and $\mathbf t'$ has more than 3 distinct terms, but $\mathbf t$ and $\mathbf t'$ are both lower mechanical words, a contradiction because of what we proved above. The same argument by contradiction proves that if $X$ has cardinality 3, then its longest element has to be a concatenation of the other two.
\end{proof}

\subsection{RC factorizations as Sturmian or 3-iet words}
Theorem~\ref{thm:main} allows to consider RC factorizations as infinite words on the \emph{finite} alphabet $X$. In this section we analyze the structure of such words.

We recall that two finite words $u,v$ are \emph{conjugate} if $u=\lambda\mu$ and $v=\mu\lambda$ for some words $\lambda,\mu$. The following characterization of Sturmian morphisms was proved in \cite[Theorem~A.1]{berthedlr}:
\begin{thm}[See~\citet{berthedlr}]
A morphism $f:A^{*}\to A^{*}$ is Sturmian if and only if it maps the three Christoffel words $01$, $001$, and $011$ to conjugates of Christoffel words.
\end{thm}
\begin{cor}
\label{thm:bdlr}
If $u$, $v$, and $uv$ are Christoffel words, then the morphisms $0\mapsto u,1\mapsto v$ and $0\mapsto v, 1\mapsto u$ are Sturmian.
\end{cor}
\begin{proof}
Immediate consequence of Proposition~\ref{thm:cmix} and the previous theorem.
\end{proof}
Let us recall a further result (\cite[Proposition~2.3.2]{ch2acow}) on the Sturmian morphisms $\varphi$ and $\tilde\varphi$ from~\eqref{eq:stmor}.
\begin{prop}[See \citet{ch2acow}]
\label{thm:phitil}
Let $\Ww$ be an infinite word.
	\begin{enumerate}
	\item If $\varphi(\Ww)$ is Sturmian, then so is $\Ww$.
	\item If $\tilde\varphi(\Ww)$ is Sturmian and $\Ww$ starts with $0$, then $\Ww$
	is Sturmian.
	\end{enumerate}
\end{prop}

Given a set $X\subseteq A^{*}$ and a word $w\in X^{*}$, by an abuse of language we say that a factorization $w=x_{1}x_{2}\cdots$, with $x_{i}\in X$ for all $i$, is a \emph{word over the alphabet $X$}, identifying it with the word $x_{1}x_{2}\cdots$ in the free monoid over $X$ (or with the word $f^{-1}(x_{1})f^{-1}(x_{2})\cdots$, where $f$ is a bijection from a new alphabet $B$ to $X$). The same identification is made also for factorizations of infinite words.

A \emph{complete return} to $v\in A^{*}$ is a finite word containing exactly two occurrences of $v$, one as a prefix and one as a suffix. If $w$ is a finite or infinite word and $v$ is a factor of $w$, then a (right) \emph{return word to $v$ in $w$} is a word $r$ such that $rv\in\Ff(w)$ is a complete return to $v$. Left returns can be defined similarly, i.e., replacing $rv$ with $vr$ in the definition. 

The following result is a known consequence of a theorem by~\citet{vuillon}  characterizing Sturmian words as the ones having exactly two return words for each factor:
\begin{prop}[See e.g.~\cite{juvuret}]
\label{thm:deriv}
Let $p$ be a prefix of a Sturmian word $\Ww$.
The sequence of return words to $p$ in $\Ww$ is Sturmian; that is:
\begin{enumerate}
\item $p$ has exactly two distinct return words in $\Ww$, say $u$ and $v$, and 
\item given the morphism $f_{p}:0\mapsto u,1\mapsto v$, the word $\Ww_{p}\in A^{\omega}$ such that $f_{p}(\Ww_{p})=\Ww$ is Sturmian.
\end{enumerate}
\end{prop} 

We can now begin to shed light on the structure of RC factorizations.
\begin{prop}
\label{thm:uv}
Let $\Ww=x_{1}x_{2}\cdots$ and $X$ be defined as in Theorem~\ref{thm:main}. Suppose
$X=\{u,v,z\}$ with $z=uv$. Then:
\begin{enumerate}
\item The factorization $\Ww=y_{1}y_{2}\cdots$, obtained from the starting RC factorization by replacing each occurrence of $z$ with $u\cdot v$  (that is, defined so that for all $i$ with $x_{i}=z$, there exists $j$ with $y_{1}\cdots y_{j-1}=x_{1}\cdots x_{i-1}$, $y_{j}=u$, and $y_{j+1}=v$) is also reversible Christoffel.
\item
The new factorization $y_{1}y_{2}\cdots$ is also a Sturmian word on the alphabet $\{u,v\}$.
\end{enumerate}
\end{prop}
\begin{proof}
The result is trivially verified when $\{u,v\}=\{0,1\}$. Let us then suppose this is not the case; by Proposition~\ref{thm:cmix} we deduce that either $u$ is a prefix of $v$, or $v$ is a suffix of $u$. Distinguishing such two cases, we shall first prove our second claim, i.e., that the new factorization defines a Sturmian word $\hat\Ww=f^{-1}(y_{1})f^{-1}(y_{2})\cdots$, where $f:0\mapsto u,1\mapsto v$.

\begin{itemize}
\item If $u$ is a prefix of $v$, let $n>0$ be the greatest integer such that $u^{n}$ is a prefix of $v$. Clearly $u^{n}$ is a prefix of $\Ww$; we shall prove that $u$ and $v$ are the return words to $u^{n}$ in $\Ww$, thus showing that $\hat\Ww$ is indeed Sturmian by Proposition~\ref{thm:deriv}.

It is easy to check that $u^{n+1}$ and $vu^{n}$ are indeed factors of $\Ww$.
Since $u$ is primitive, it is clearly a return word to $u^{n}$. Let $v=u^{n}u'$ for some $u'\in A^{*}$. By Proposition~\ref{thm:cmix}, $u'$ and $uu'$ are Christoffel words, so that we have either $\{u,u'\}=\{a,b\}$ or $u'\in\Su(u)$ ($u$ cannot be a prefix of $u'$ by the maximality of $n$).
Clearly $u^{n}$ is a prefix and a suffix of $vu^{n}=u^{n}u'u^{n}$; we need to show that  has no other occurrences. If $\{u,u'\}=A$, this is trivial. If $u'$ is a suffix of $u$, then $u^{n}$ cannot have internal occurrences in $u^{n}u'u^{n}$ because $u$ is unbordered. Therefore $u$ and $v$ are return words to $u^{n}$ in $\Ww$, so that $\hat\Ww$ is Sturmian.

\item If $v$ is a suffix of $u$, let $m>0$ be the greatest integer such that $v^{m}$ is a suffix of $u$. The same argument as above shows that $v^{m}u$ and $v^{m+1}$ are complete returns to $v^{m}$. Thus, writing $u=v'v^{m}$ for some $v'\in A^{*}$, we get that $v$ and $v^{m}v'$ are the return words to $v^{m}$ in the word $v^{m}\Ww$. The sequence of these return words is again determined by $\hat\Ww$, as $v^{m}\Ww=g(\hat\Ww)$ with $g:0\mapsto v^{m}v', 1\mapsto v$. Hence, to prove that $\hat\Ww$ is Sturmian, by Proposition~\ref{thm:deriv} we only need to show that $v^{m}\Ww$ is Sturmian.

By mirroring the argument used in the proof of Proposition~\ref{thm:abcom}, we get that since $\{\Ww,\Ww'\}\neq\{0\Cc,1\Cc\}$, there exist arbitrarily long Abelian equivalent words $r,r'$ such that $r\Ww$ and $r'\Ww'$ are Sturmian. By Theorem~\ref{thm:main}, the set of terms in the RC factorization of $r\Ww$ determined by $r'\Ww'$ cannot be larger than $X=\{u,v,z\}$. Since $u=v'v^{m}$ and $z=v'v^{m+1}$, any sufficiently long $r$ will have $v^{m}$ as a suffix. We have thus proved that $v^{m}\Ww=g(\hat\Ww)$ is Sturmian, and so is $\hat\Ww$.
\end{itemize}
Since $\hat\Ww$ is Sturmian, and the morphism $\tilde f: 0\mapsto\tilde u, b\mapsto\tilde v$ is Sturmian by Corollary~\ref{thm:bdlr}, the word $\tilde f(\hat\Ww)$ is Sturmian too, so that the new factorization is actually the RC factorization of $\Ww$ determined by $\tilde f(\hat\Ww)$. This completes the proof.
\end{proof}

\begin{remark}
If $u$, $v$, and $uv$ are the terms of an RC factorization $\Ww=x_{1}x_{2}\cdots$, then by Proposition~\ref{thm:cmix} exactly one among 
$u^{2}v$ and $uv^{2}$ is a factor of $\Ww$. As a consequence of Proposition~\ref{thm:uv}, it is easy to see that if $u^{2}v$ (resp.~$uv^{2}$) is a factor of $\Ww$, then each occurrence of $v$ (resp.~$u$) in the factorization is preceded by $u$ (resp.~followed by $v$), provided that $x_{1}\neq v$. This gives rise to the following ``converse'' of Proposition~\ref{thm:uv}.

\end{remark}

\begin{prop}
Let $u$, $v$, and $\Ww=x_{1}x_{2}\cdots$ be as above. Replacing in the factorization each occurrence of $u\cdot v$ with one of $z=uv$ (that is, defining $y_{n}$ for $n>0$ so that for all $i$ where $x_{i}=u$ and $x_{i+1}=v$ there exists $j$ with $x_{1}\cdots x_{i-1}=y_{1}\cdots y_{j-1}$ and $y_{j}=z$) produces a new RC factorization of $\Ww$, which is also a Sturmian word (over $\{z,u\}$ or $\{z,v\}$).
\end{prop}
\begin{proof}
Let us first assume $u^{2}v\in\Ff(\Ww)$.
By the previous Remark, it is clear that the replacement yields a factorization of $\Ww$ where $v$ does not appear.
Hence, the new factorization can be seen as an infinite word on the alphabet $\{u,z\}$, i.e., the image of some word $\bar\Ww\in A^{\omega}$ under the morphism $h:0\mapsto z, 1\mapsto u$. To see that $\bar\Ww$ is Sturmian, by Proposition~\ref{thm:phitil} it suffices to observe that $\varphi(\bar\Ww)=\hat\Ww$, where
$\hat\Ww$ is given by Proposition~\ref{thm:uv} and $\varphi$ is the Fibonacci morphism as in~\eqref{eq:stmor}.

To show that this new factorization is actually RC, once again we just need to observe that it is obtained by Abelian comparison with the word $\tilde h(\bar\Ww)$, where the morphism $\tilde h:0\mapsto\tilde z, 1\mapsto\tilde u$ is Sturmian by Corollary~\ref{thm:bdlr}.

Now suppose $uv^{2}\in\Ff(\Ww)$ instead. Essentially the same argument as above applies; we consider the new factorization as the image of an infinite word $\bar\Ww$ under the morphism $h:0\mapsto z, 1\mapsto v$. We have $(E\circ\tilde\varphi)(\bar\Ww)=\hat\Ww$, and $\bar\Ww$ starts with 0, so that $\bar\Ww$ is Sturmian by Proposition~\ref{thm:phitil}; the factorization is RC since it is obtained by Abelian comparison with $\tilde h(\bar\Ww)$, where $\tilde h:0\mapsto\tilde z, 1\mapsto\tilde v$ is Sturmian by Corollary~\ref{thm:bdlr}.
\end{proof}

We recall (cf.~\cite{ArBeMaPe}) that a \emph{3-iet word} is an infinite word coding the orbit of a point $\rho$ under a \emph{three-interval exchange transformation} $T$. More precisely, given an interval $I=[0,\ell]\subseteq\mathbb R$ containing $\rho$ and subdivided in three intervals $I_{a}=[0,\alpha]$, $I_{b}=[\alpha,\alpha+\beta]$, and $I_{c}=[\alpha+\beta,\ell]$, we let $T$ be the piecewise linear transformation of $I$ exchanging the three subintervals according to the permutation $(321)$, i.e., let $T:I\to I$ be defined by $T(\xi)=\xi+t_{x}$ if $\xi\in I_{x}$, where $x\in\{a,b,c\}$ and
\[t_{a}=	\ell-\alpha,\quad t_{b}=\ell-2\alpha-\beta,\quad\text{ and }
t_{c}=-\alpha-\beta\;.\]
The 3-iet word determined by $\alpha$, $\beta$, $\ell$, and $\rho$ is then the infinite word indexed over $\Nn$ whose $n$-th letter is $x\in\{a,b,c\}$ if $T^{n}(\rho)\in I_{x}$.

Let $\sigma,\sigma':\{a,b,c\}^{*}\to A^{*}$ be morphisms defined by
\begin{equation}
\label{eq:sigma}
\sigma(a)=0=\sigma'(a),\quad \sigma(b)=01,\;\sigma'(b)=10,\quad \sigma(c)=1
=\sigma'(c)\;.
\end{equation}
The following result was proved in \cite[Theorem~A]{ArBeMaPe}:
\begin{thm}[See~\citet{ArBeMaPe}]
\label{thm:arbemape}
An infinite word $\mathbf u$ on the alphabet $\{a,b,c\}$, whose letters have positive frequencies, is an aperiodic 3-iet word if and only if $\sigma(\mathbf u)$ and $\sigma'(\mathbf u)$ are Sturmian words.
\end{thm}

As a consequence, we get that RC factorizations are in general 3-iet words:
\begin{cor}
\label{thm:rc3iet}
Let $\Ww$ be a Sturmian word, and $\Ww=x_{1}x_{2}\cdots$ be an RC factorization with $X=\{x_{n}\mid n>0\}=\{u,v,z\}$, $z=uv$. If every word of $X$ occurs more than once in the factorization, then $x_{1}x_{2}\cdots$ is an aperiodic 3-iet word over the alphabet $\{u,z,v\}$.
\end{cor}
\begin{proof}
Let $\tau:\{a,b,c\}\to\{u,v,z\}$ be defined by $\tau(a)=u$, $\tau(b)=z$, and $\tau(c)=v$. We need to show that the infinite word
$\dot\Ww:=\tau^{-1}(x_{1})\tau^{-1}(x_{2})\cdots$ is a 3-iet word. 
Clearly, exchanging the roles of $\Ww$ and $\Ww'$ (and letting $\tau(a)=\tilde u$ etc.) does not change $\dot\Ww$; therefore, without loss of generality, in the following we can assume by Proposition~\ref{thm:uplow} that $X$ is made of lower Christoffel words.


Let $\alpha$ be the slope of $\Ww$ and $\Ww'$, and let $\rho$ and $\rho'$ respectively be their intercepts. As is well-known, any factor $\gamma$ of $\Ww$ corresponds to an interval
$I_\gamma$ on the unit circle, i.e., $\gamma$ occurs at position $n$ in $\Ww$ if and only if $\{n\alpha+\rho\}\in I_{\gamma}$; moreover, since $\Ww$ and $\Ww'$ have the same slope, the positions of $\tilde\gamma$ in $\Ww'$ are identified by
\[\{n\alpha+\rho'\}\in I_{\tilde\gamma}\iff \{n\alpha+\rho\}\in
I_{\tilde\gamma}-\rho'+\rho\]
where $I_{\tilde\gamma}-\rho'+\rho$ is a translation on the unit circle (i.e., the sum is taken modulo 1).

Let now $\gamma\in X$, and suppose first that $|\gamma|>1$.
The relation
\begin{equation}
\label{eq:occur}
\{n\alpha+\rho\}\in I_{\gamma}\cap(I_{\wt\gamma}-\rho'+\rho)
\end{equation}
identifies the positions $n$ of all occurrences of $\gamma$ in $\Ww$ such that $\tilde\gamma$ occurs at the same position in $\Ww'$. As $|\gamma|>1$, we have $\gamma=0q1$ and $\tilde\gamma=1q0$ for some word $q$. We claim that all occurrences of $\gamma$ whose position satisfies \eqref{eq:occur} appear in the RC factorization. Indeed, if one of such occurrences did not correspond to
$x_{i}$ for any $i\geq 1$, then the first $0$ of $0q1$ (resp.~$1$ of $1q0$) would have to be the last letter of some $x_{j}$ (resp.~$\wt{x_{j}}$), against the fact that all $x_{j}$'s are lower Christoffel words. 

Now suppose $\gamma$ is a letter. Without loss of generality, we can assume
$\alpha<1/2$. If $\gamma=1$, then necessarily $X=\{0,1,01\}$, so that~\eqref{eq:occur} again identifies exactly all occurrences of $\gamma$ in $\Ww$ that appear in the RC factorization. Let then $\gamma=0$, so that $X$ 
is $\{0,0^{n}1,0^{n+1}1\}$ for some $n\geq 0$. As a consequence of the balance condition, it is easy to see that a position $n$ where $0$ occurs in both $\Ww$ and $\Ww'$ corresponds to $x_{i}$ for some $i\geq 1$ if and only if it is followed in $\Ww$ by $x_{i+1}=0^{n}1$. Hence, such positions are exactly those that satisfy
\begin{equation}
\label{eq:occ0}
\{n\alpha+\rho\}\in I_{0^{n+1}1}\cap(I_{010^{n}}-\rho'+\rho)\;.
\end{equation}

In all cases and for all $\gamma\in X$, we have identified intervals corresponding to the occurrences of $\gamma$ in the RC factorization, namely the ones in~\eqref{eq:occur} or~\eqref{eq:occ0}. By hypothesis, $\gamma$ occurs at least twice in the factorization; by the irrationality of $\alpha$, such intervals must then have nonempty interior, so that the gaps between consecutive integers $n$ satisfying~\eqref{eq:occur} or~\eqref{eq:occ0} are bounded. Thus, every term in the RC factorization, and so every letter in $\dot\Ww$, occurs with positive frequency.

By Theorem~\ref{thm:arbemape}, it remains to prove that the two words $\sigma(\dot\Ww)$, $\sigma'(\dot\Ww)$ are Sturmian, where $\sigma$ and $\sigma'$ are the morphisms defined in~\eqref{eq:sigma}.
In fact, it is easy to check that $\sigma(\dot\Ww)=\hat\Ww$ and
$\sigma'(\dot\Ww)=\hat\Ww'$, i.e., the words $\sigma(\dot\Ww)$ and $\sigma'(\dot\Ww)$ coincide with the Sturmian words obtained in the proof of Proposition~\ref{thm:uv}, applied respectively to the RC factorizations $\Ww=x_{1}x_{2}\cdots$ and
$\Ww'=\widetilde{x_{1}}\widetilde{x_{2}}\cdots$. The result follows.
\end{proof}
\begin{remark}
The hypothesis that every word of $X$ occurs more than once in the factorization is necessary. For instance, it is easy to see that for the RC factorization $x_{1}x_{2}\cdots x_{n}\cdots$ of $\Ww=010\mathbf f$ determined by $\Ww'=1001\mathbf f$, one has $X=\{0,01,001\}$, but the term $0$ occurs exactly once, as $x_{2}$. In fact, by Proposition~\ref{thm:uv}, the word $\mathbf f=x_{3}x_{4}\cdots$ is Sturmian over the alphabet $\{01,001\}$.
\end{remark}

\section{Future work}

We believe that much can still be said about reversible Christoffel factorizations; for example, it would be interesting to characterize the set of terms $X$ in terms of the slope $\alpha$ and the difference between the intercepts, in particular to distinguish when the cardinality of $X$ is 2 or 3.
\subsection*{Acknowledgment}
The second author would like to thank the Department of Mathematics at the University of Turku for its support during his part-time employment there in 2011--2012 as a member of the FiDiPro unit.
\bibliographystyle{model1-num-names}

\end{document}